\documentclass{article}

\usepackage[T1]{fontenc}
\usepackage{graphicx}

\usepackage[utf8]{inputenc}
\usepackage{amsmath,amssymb,amsthm}
\usepackage{url,hyperref}
\usepackage{fancyvrb}
\usepackage{algorithm,algorithmicx,algpseudocode}

\usepackage[margin=1in]{geometry}

\newcommand{\eps}{\epsilon}

\newcommand{\etal}{{\it{et al.}}}

\newcommand{\E}{\mathcal{E}}

\newcommand{\EE}{\mathbb{E}}
\newcommand{\var}{\mathrm{Var}}

\newcommand{\C}{\mathcal{C}}

\newcommand{\ti}[1]{\tilde{#1}}
\newcommand{\wt}[1]{\widetilde{#1}}

\newtheorem{claim}{Claim}[section]

\newtheorem{thm}{Theorem}[section]
\newtheorem{lem}{Lemma}[section]
\newtheorem{fact}{Fact}
\newtheorem{defn}{Definition}[section]

\allowdisplaybreaks

\usepackage{color}

\begin{document}

\title{Improved Sublinear-time Moment Estimation using Weighted Sampling} 


\author{
Anup Bhattacharya\thanks{Authors are ordered alphabetically}\\ NISER, Bhubaneswar, India \\ anup@niser.ac.in \and
Pinki Pradhan \\ NISER, Bhubaneswar, India \\ pinki.pradhan@niser.ac.in}

\maketitle

\begin{abstract}

In this work we study the {\it moment estimation} problem using weighted sampling. Given sample access to a set $A$ with $n$ weighted elements, and a parameter $t>0$, we estimate the $t$-th moment of $A$ given as $S_t=\sum_{a\in A} w(a)^t$. For $t=1$, this is the {\it sum estimation} problem for which sublinear time algorithms are known. The moment estimation problem along with a number of its variants have been extensively studied in streaming, sublinear and distributed communication models. Despite being well studied, we don't yet have a complete understanding of the sample complexity of the moment estimation problem in the sublinear model and in this work, we make progress on this front. On the algorithmic side, our upper bounds match the known upper bounds for the problem for $t>1$. To the best of our knowledge, no sublinear algorithms were known for this problem for $0<t<1$. We design a sublinear algorithm for this problem for $t>1/2$ and show that no sublinear algorithms exist for $t\leq 1/2$. We prove a $\Omega(\frac{n^{1-1/t}\ln 1/\delta}{\eps^2})$ lower bound for moment estimation for $t>1$, and show optimal sample complexity bound $\Theta(\frac{n^{1-1/t}\ln 1/\delta}{\eps^2})$ for moment estimation for $t\geq 2$. Hence, we obtain a complete understanding of the sample complexity for moment estimation using proportional sampling for $t\geq 2$. We also study the moment estimation problem in the beyond worst-case analysis paradigm and identify a new {\it moment-density} parameter of the input that characterizes the sample complexity of the problem using proportional sampling and derive tight sample complexity bounds with respect to that parameter. We also study the moment estimation problem in the {\it hybrid sampling} framework in which one is given additional access to a uniform sampling oracle. We show access to a hybrid sampling framework does not provide any additional gain for this problem over a proportional sampling oracle in the worst case. 

\end{abstract}

\section{Introduction}


Let $A$ be any weighted set of $n$ elements. Each element in $A$ is associated with a weight using the function $w:A\rightarrow [0,\infty)$. Given an input parameter $t>0$, the problem of estimating the $t$-th moment of $A$, expressed as $S_t=\sum_{a\in A} w(a)^t$, is called the {\it moment estimation} problem\footnote{The $t$-th moment of $A$ is also defined as $S_t=1/n \cdot \sum_{a\in A} w(a)^t$. We follow $S_t=\sum_{a\in A} w(a)^t$ for simplicity of calculations \cite{ERS2019}.}. In this work we study the moment estimation problem in a model in which we don't have direct access to the weights of the elements in $A$, instead we get indirect access to the weights using samples from an oracle. Our objective is to obtain a good approximation of the $t$th moment $S_t$ of $A$ while making a small number of queries to the oracle. 


Estimation tasks on very large datasets are often performed using samples from the dataset. One common objective in these settings is to compute an approximate answer using only a small number of samples. Sublinear algorithms are algorithms that access only a tiny portion of the data (sublinear in the input size) and compute an approximate answer. For many problems of interest, uniform sampling-based algorithms require a lot of samples, and hence might not be very useful in designing sublinear algorithms where as weighted sampling-based algorithms often give better performance guarantees. This is the case, for example, for the sum estimation problem. Uniform sampling-based approaches require $\Omega(n)$ samples for the sum estimation problem in the worst-case where as weighted sampling-based algorithms require only $O(\sqrt{n}/\eps)$ samples \cite{MPX2007,BT2022}. This, however, uses the assumption of access to a stronger query oracle that returns samples according to a weighted distribution. In applications where access to such a stronger query oracle is available, one might use the above algorithms with better performance guarantees. One such use case is the parameter estimation problems on graphs where one is given access to a random edge sampling oracle. Details about using weighted sampling for graph parameter estimation problems and other application areas can be found in \cite{MPX2007,ABGPRY2018,BT2022}.

In this work we study the sample complexity of estimators for the moment estimation problem assuming access to a proportional sampling oracle on the set $A$. Proportional sampling on set $A$ returns an element $a\in A$ with probability proportional to $w(a)$. Let $W=\sum_{a\in A} w(a)$. Using proportional sampling an element $a\in A$ is chosen with probability $w(a)/W$. Next, we formally define the moment estimation problem that we study in this work.

\begin{defn}$((\eps,\delta)-\text{Moment Estimation})$ Given sample access to a set $A$ of $n$ weighted elements and input parameters $\eps,\delta\in (0,1)$, $t>0$, design an algorithm $ALG$ that returns $ALG(A,t,\eps,\delta)$ such that $$\Pr[(1-\eps)S_t\leq ALG(A,t,\eps,\delta)\leq (1+\eps)S_t]\geq 1-\delta$$ \end{defn}


The moment estimation problem is one of the fundamental problems with a number of variants that are studied in the literature of sublinear algorithms. For $t=1$, this is the {\it sum estimation} problem. Motwani~\etal~\cite{MPX2007} initiated the study of designing sublinear algorithms for the sum estimation problem assuming access to a proportional sampling oracle. Recently, Beretta and T{\v{e}}tek \cite{BT2022} have improved the sample complexity of the sum estimation problem to $O(\sqrt{n}/\eps)$ using proportional sampling. Estimating frequency moments in streams is one of the most well-studied problems in the streaming literature \cite{AMS1999}. Moment estimation is also well studied in the distributed communication models \cite{JW2023}. This problem can also be thought of as the vector norm estimation problem where the weights correspond to the entries of a vector. When the set of elements corresponds to the set of vertices in a graph and the weights correspond to the degrees of the vertices in the graph, this problem is known as the {\it degree distribution moment estimation} problem. For $t=1$, this is the edge estimation problem in graphs for which Feige \cite{F2006} and Goldreich and Ron \cite{GR2008} designed sublinear algorithms. For the related degree distribution moment estimation problem in graphs, the authors in \cite{ERS2018,ERS2019,GRS2011} designed sublinear time algorithms in the sparse graph model in which one is allowed to make uniform vertex, degree and neighbour queries. Aliakbarpour~\etal~\cite{ABGPRY2018} designed a sublinear algorithm for estimating $\sum_{a\in A} {w(a)\choose t}$ using $O(\frac{n^{1-1/t}\ln 1/\delta}{\eps^2})$ queries in the sparse graph model with additional access to a random edge oracle, where $w(a)$ denotes the degree of vertex $a$. Various applications of the moment estimation problem including its variants are discussed in \cite{ERS2018,ABGPRY2018,BT2022}.

A number of recent works \cite{AN2022,TT2022,G2017} highlighted the importance of designing sublinear algorithms with optimal dependence not only on $n$ but also on $\eps$ and $\delta$. Beretta and T{\v{e}}tek have improved the sample complexity of the sum estimation problem from $\tilde{O}(\sqrt{n}/\eps^{7/2})$ proportional samples in Motwani~\etal~\cite{MPX2007} to $\Theta(\sqrt{n}/\eps)$ proportional samples \cite{BT2022}. Assadi and Nguyen designed a sublinear algorithm to estimate $h$-index with optimal sample complexity dependence on all input parameters \cite{AN2022}. Despite a lot of work on sublinear algorithms for the moment estimation problem, we don't have a complete understanding of the sample complexity of the moment estimation problem using weighted sampling. One of the primary motivations of this work is to make progress in our understanding of the moment estimation problem in sublinear models. Next, we discuss the main results of this work.


\subsection{Our contributions}


\subsubsection{Moment estimation using proportional sampling} 

Our results for the $(\eps,\delta)$-moment estimation problem using proportional sampling for $t>1$ are as follows. 

\begin{thm} There exists an algorithm $ALG$ that given proportional sampling access to the weights of the elements of a set $A$ and parameters $t>1$, $\eps,\delta\in (0,1)$, provides an $(\eps,\delta)$-estimate of $S_t$ using $O((\frac{\sqrt{n}}{\eps}+\frac{n^{1-1/t}}{\eps^2}) \ln\frac{1}{\delta})$ samples. \end{thm}

For $t\geq 2$, the sample complexity of our algorithm is $O(\frac{n^{1-1/t}\ln\frac{1}{\delta}}{\eps^2})$. Our next result shows that this bound is tight.

\begin{thm} For any $\eps,\delta\in (0,1)$ and $t>1$, any randomized algorithm that computes an $(\eps,\delta)$-estimate of $S_t$ requires $\Omega(\frac{n^{1-1/t}\ln\frac{1}{\delta}}{\eps^2})$ proportional samples. \end{thm}

For $0<t<1$, we show the following. We design an $(\eps,\delta)$-estimator for $S_t$ when $t>1/2$.

\begin{thm} There exists an algorithm $ALG$ that given proportional sampling access to the weights of the elements of a set $A$ and parameters $1/2<t<1$, $\eps,\delta \in (0,1)$, provides an $(\eps,\delta)$-estimate of $S_t$ using $O((\frac{\sqrt{n}}{\eps}+\frac{n^{\frac{1}{t}-1}}{\eps^2})\ln \frac{1}{\delta})$ samples. \end{thm}

We show that when $t\leq 1/2$, no sublinear algorithms exist for this problem. More specifically, we show the following.

\begin{thm} For any $\eps>0$ and $t\leq 1/2$, any randomized algorithm that computes an $(\eps,1/3)$-estimate of $S_t$ requires $\Omega(n)$ proportional samples. \end{thm}

\paragraph*{Significance of our results and comparisons with known results} 
\begin{itemize}

	\item {\bf Case $t>1$}: Eden~\etal~\cite{ERS2018,ERS2019} designed sublinear algorithms for the degree distribution moment estimation problem in the sparse graph model assuming access to uniform vertex, degree and neighbour queries. Aliakbarpour~\etal~\cite{ABGPRY2018} improved the sample complexity bound for this problem assuming additional access to a random edge oracle. Aliakbarpour~\etal~observed that estimating the number of $t$-stars in a graph can be seen as a variant of the moment estimation problem. Let $w(a)$ denote the degree of vertex $a\in A$. Then, $\sum_{a \in A}\binom {w(a)}t$ counts the number of $t$-stars in the graph. The authors designed sublinear algorithms for this problem assuming access to a random edge sampling oracle in the sparse graph model. To the best of our knowledge, this problem of estimating $\sum_{a\in A} {w(a)\choose t}$ of Aliakbarpour~\etal~seems to be the most closely related problem to ours. In our setting, the result of Aliakbarpour~\etal~gives a sublinear-time algorithm for moment estimation using $O(\frac{n^{1-1/t}\ln 1/\delta}{\eps^2})$ proportional samples.\footnote{The authors stated their bound as $O(n^{1-1/t}/\eps^3)$ for constant success probability, but we believe using \cite{BT2022}, it can be improved to $O(n^{1-1/t}/\eps^2)$.} They also proved a $\Omega(n^{1-1/t})$ sample complexity lower bound for this problem. In the following we discuss the key differences of our results with theirs.

\begin{enumerate}
	\item Upper bound for moment estimation: We note that for $t>1$ our upper bound matches with that of Aliakbarpour~\etal~\cite{ABGPRY2018} and the algorithmic ideas and the analysis of these works are similar. We do not claim much technical novelty for our upper bound. However, we point out that our setup is strictly more general than that of Aliakbarpour~\etal~and their result can be recovered in our setup. We also note that our algorithm uses only proportional samples where as the algorithm of Aliakbarpour~\etal~uses queries in the sparse graph model with additional access to random edge samples. Since it is known that using uniform edge samples one can sample vertices with probabilities proportional to their degrees, it appears that access to proportional samples enables one to design sublinear algorithms for the moment estimation problem.
	\item Lower bound for moment estimation: Aliakbarpour~\etal~gave a $\Omega(n^{1-1/t})$ lower bound for this problem. We show an improved lower bound of $\Omega(\frac{n^{1-1/t}\ln\frac{1}{\delta}}{\eps^2})$ for $t>1$. For $t\geq 2$, this settles the sample complexity of $\Theta(\frac{n^{1-1/t}\ln\frac{1}{\delta}}{\eps^2})$ for the moment estimation problem with optimal dependence on all input parameters $n,\eps,\delta$.
\end{enumerate}

	\item {\bf Case $t<1$}: Moment estimation for $0<t<1$ is a well motivated problem with a number of applications \cite{JW2023}. However, to the best of our knowledge, no sublinear algorithms were known for the moment estimation problem when $t<1$. The moment estimation algorithms of Eden~\etal~\cite{ERS2018,ERS2019} and Aliakbarpour~\etal~\cite{ABGPRY2018} work only for $t\geq 1$. We design the first sublinear algorithm for the moment estimation problem for $t>1/2$. We also show that for $t\leq 1/2$, no sublinear algorithms exist for this problem. 
\end{itemize}

\subsubsection{Characterization of Sample Complexity}

The moment estimation problem using proportional sampling requires $\Omega(\frac{n^{1-1/t}\ln\frac{1}{\delta}}{\eps^2})$ samples in the worst case. We study the moment estimation problem in a beyond worst-case analysis paradigm and identify a parameter of the input that characterizes the sample complexity of the problem using proportional sampling. For the degree distribution moment estimation problem in graphs, Eden~\etal~\cite{ERS2019} identified the {\it arboricity} of a graph as the relevant parameter and obtained sample complexity bounds in terms of the arboricity of the graph.

We introduce a new \textit{moment-density} parameter $\rho$ of the input that characterizes the sample complexity for the moment estimation problem. For $W=\sum_{a\in A} w(a)$ and $S_t=\sum_{a\in A} w(a)^t$, we define the parameter $\rho$ as $$\rho=\max_{L\subseteq A} \frac{\frac{\sum_{a\in L} w(a)^t}{\sum_{a\in L} w(a)}}{\frac{\sum_{a\in A} w(a)^t}{\sum_{a\in A} w(a)}}=\max_{L\subseteq A} \frac{\sum_{a\in L} w(a)^t}{\sum_{a\in L} w(a)} \cdot \frac{W}{S_t}$$ 

The motivation for writing the moment-density parameter $\rho$ in the above form is as follows. The key idea behind the lower bound instances for the moment estimation problem using proportional sampling (described in Section \ref{sec:lower-proportional}) is to assign large weights on only a few input elements such that these elements are not easily detected using proportional sampling but the $t$th power of their weights dominate the moment value of the instance. This is why proportional sampling requires a lot of samples on these kind of instances. Now, if we were allowed to sample elements of $A$ with probabilities proportional to their $t$th power of the weights, then the complexity of the moment estimation problem becomes the same as the complexity of the sum estimation problem using proportional sampling. But, when we sample using proportional sampling, there might be some elements with slightly larger weights having outsized influence in the moment value of the instance. The moment-density parameter $\rho$ is defined such that its value will be large on those instances. For any subset $L\subseteq A$, let $\rho_L$ denote the ratio of the fractional contribution of the elements in $L$ to the moment value and the probability that an element in $L$ is going to be sampled using proportional sampling. We have $\rho=\max_{L\subseteq A} \rho_L$. Alternatively, $\rho$ captures the maximum contribution of any subset $L$ to the $t$-th moment $S_t$ relative to the sum of the weights of elements in $L$. For scaling we divide it by $S_t/W$.

We show an upper bound for $(\eps,\delta)$-estimate of the moment using proportional sampling where we write the sample complexity in terms of $\eps,\delta$ and $\rho$.

\begin{thm} There exists an algorithm $ALG$ that given proportional sampling access to the weights of the elements in a set $A$ with moment-density parameter $\rho$ and parameters $\eps,\delta\in (0,1), t>1$, provides an $(\eps,\delta)$-estimate of $S_t$ using $O((\sqrt{n}/\eps+\frac{\rho}{\eps^2})\ln 1/\delta)$ samples. \end{thm}

Next, we show an almost tight lower bound on the sample complexity for the moment estimation problem on instances with moment-density parameter $\rho$ in terms of $\rho,\eps,\delta$.

\begin{thm} For any $\eps,\delta\in (0,1)$ and $t>1$, any randomized algorithm for $(\eps,\delta)$-estimate of $S_t$ on an instance with moment-density parameter $\rho$ requires $\Omega(\frac{\rho \ln 1/\delta}{\eps })$ proportional samples. \end{thm}

\subsubsection{Moment Estimation using Hybrid Sampling}

Proportional sampling-based algorithms for the moment estimation problem require $\Omega(\frac{n^{1-1/t}\ln1/\delta}{\eps^2})$ samples for $t>1$ in the worst case. One natural idea to design algorithms with improved sample complexity bounds is to give more power to the algorithm designer in the form of access to a stronger query oracle. In this section we explore whether additional access to a uniform sampling oracle allows one to design an algorithm with improved sample complexity for this problem. The hybrid sampling framework allows one to use both proportional and uniform samples. This study is motivated by the fact that, for the sum estimation problem, Motwani~\etal~\cite{MPX2007} and Beretta and T{\v{e}}tek \cite{BT2022} exploited access to a hybrid sampling oracle to design algorithms that make $\tilde{O}(n^{1/3}/\eps^{9/2})$ and $O(n^{1/3}/\eps^{4/3})$ samples, respectively. This is in contrast to the $\Theta(\sqrt{n}/\eps)$ sample complexity bound for the sum estimation problem using only proportional samples. In this work we explore whether hybrid sampling might give us better sample complexity bounds for the moment estimation problem. We prove a lower bound result showing that no improved algorithm with better sample complexity bounds exists for the moment estimation problem using hybrid sampling. We state the result next and prove it in Section \ref{sec:lower-hybrid}.

\begin{thm}(Lower bound using hybrid sampling) For any $\eps,\delta \in (0,1)$ and $t>1$, any algorithm having access to a hybrid sampling oracle requires at least $\Omega(\frac{n^{1-1/t}\ln 1/\delta}{\eps^2})$ samples to compute an $(\eps,\delta)$-estimate for $S_t$. \end{thm}

\subsection{Techniques} The main idea behind the upper bound is fairly standard and we describe it as follows. Suppose we obtain a sample $a\in A$ of weight $w(a)$ using proportional sampling. Let us set $X=w(a)^t$. Then, $\EE[X]=\sum_{a\in A} w(a)^t p_a$, where $p_a=\frac{w(a)}{W}$ denotes the probability of sampling the element $a$ and $W=\sum_{a\in A} w(a)$. Suppose we know $p_a$ for each $a\in A$, then $X=w(a)^t/p_a$ would give us an unbiased estimator of $S_t$. However, we don't know $W$ and hence don't know about these probabilities $p_a$. Our crucial observation here is to use the sum estimation algorithm in Beretta and T{\v{e}}tek \cite{BT2022} to obtain an $(\eps_1,\delta/2)$ estimate $\wt{W}$ of $W$, and use this in turn to obtain estimates $\tilde{p}_a$ for $p_a$. Our estimator built in this manner would not be an unbiased estimator of $S_t$ but we can reduce the bias of the estimator considerably by choosing an appropriate $\eps_1$. For our algorithmic results, we use $\eps_1=\eps/2$. We use similar techniques for designing algorithms in related settings. For the lower bound we use Yao's minimax lemma to construct families of instances that cannot be distinguished using a small number of proportional samples. A number of lower bound results were known for variants of the moment estimation problem \cite{GRS2011,ERS2018,ABGPRY2018}. Our lower bound constructions are motivated from these lower bound constructions. Unlike the earlier lower bounds though, our lower bounds have optimal dependence in $n,\eps,\delta$. 

\subsection{Related Works} Motwani~\etal~\cite{MPX2007} initiated the study of the sum estimation problem using access to a proportional sampling oracle and designed the first sublinear algorithm for this problem that uses $\tilde{O}(\sqrt{n})$ queries. They also designed sublinear algorithms in the hybrid sampling framework using $\tilde{O}(n^{1/3})$ queries. Beretta and T{\v{e}}tek \cite{BT2022} have recently improved these results; they prove $\Theta(\frac{\sqrt{n}}{\eps})$ sample complexity bound using proportional sampling and in the hybrid sampling setting gives almost tight sample complexity bound of $O(n^{1/3}/\eps^{4/3})$. Variants of these problems are also studied in the graph parameter estimation literature, where the elements of the set correspond to the vertices of a graph and the weights correspond to the degrees of the vertices and we are given query access to the graph. The sum estimation problem in this context becomes the {\it edge estimation} problem and moment estimation is known as the degree distribution moment estimation problem. Eden~\etal~\cite{ERS2018,ERS2019} studied the {\it degree distribution moment estimation} problem in the graph query model and designed sublinear algorithms with improved sample complexity bounds. Aliakbarpour~\etal~\cite{ABGPRY2018} studied the estimation of the number of $t$-stars in a graph and showed this problem to be closely related to the moment estimation problem. Moment estimation problem is also studied in streaming and distributed communication models \cite{AMS1999,JW2023}.


\section{Preliminaries} We use the following well known facts in the analysis.

\begin{fact}\label{fact:norms} For any vector $x\in \C^n$, and for any $0<r<p$, we have $||x||_p\leq ||x||_r\leq n^{(1/r-1/p)} ||x||_p$. \end{fact}

\begin{lem}(Chernoff bound \cite{MU2017})\label{lem:chernoff} Let $Z_1,Z_2,\ldots,Z_v$ be independent and identically distributed Bernoulli random variables. Let $Z=\sum_{i=1}^v Z_i$. Then, $\Pr[Z\leq(1-\gamma) \EE[Z]]\leq e^{-1/2 \cdot \gamma^2 \cdot \EE[Z]}$. \end{lem}

\section{Estimation of Moments using Proportional Sampling}\label{sec:moments}

We describe our algorithm for the moment estimation problem using proportional sampling. We mentioned earlier that for $t>1$, our upper bounds match those of Aliakbarpour~\etal~\cite{ABGPRY2018}. However, since our algorithm works in strictly more general settings and the algorithm for $1/2< t<1$ uses the same ideas, we describe it in detail. Let $A$ be a set of $n$ weighted elements. We assume access to a proportional sampling oracle on the weights of the elements in $A$. For a proportional sample, the oracle returns an element $a_j\in A$ with probability $w(a_j)/W$, where $W=\sum_{a_j\in A} w(a_j)$. Given parameters $t>1,\eps,\delta\in (0,1)$, we design an $(\eps,\delta)$-estimate of $S_t=\sum_{a_j\in A} w(a_j)^t$.

\begin{thm} There exists an algorithm $ALG$ that given proportional sampling access to the weights of the elements in a set $A$ and parameters $t>1,\eps,\delta\in (0,1)$, provides an $(\eps,\delta)$-estimate of $S_t$ using $O(\frac{\sqrt{n}\log 1/\delta}{\eps} + \frac{n^{1-1/t} \log 1/\delta}{\eps^2})$ samples. \end{thm}

\begin{algorithm}[H]
    \caption{Moment Estimation using Proportional Sampling}
    \label{alg:estmoments}
    \begin{algorithmic}[1] 
        \Procedure{MomentEstimator}{$A,t,\eps,\delta$} 
                \State Let $\wt{W}$ denote an $(\eps_1=\eps/2,\delta/2)$-estimate of $W$ using the sum estimation algorithm of \cite{BT2022}. This step requires $480\cdot \frac{\sqrt{n}\log (2/\delta)}{\eps}$ proportional samples.
                \For $~r=1$ to $v = 48 \cdot \log 2/\delta$  
                        \For $~j=1$ to $l=48 \cdot n^{1-1/t}/\eps^2$
                                \State Let $a_j$ denote a proportional sample of weight $w(a_j)$. 
				\State Compute $\ti{p}_j=\frac{w(a_j)}{\wt{W}}$.
                                \State Set $X_j=\frac{w(a_j)^t}{\ti{p}_j}$
                        \EndFor
                        \State $Y_r = \frac{\sum_{j=1}^l X_j}{l}$
                \EndFor 
                \State \textbf{return} median$(Y_1,\ldots,Y_v)$
        \EndProcedure
    \end{algorithmic}
\end{algorithm}

Algorithm \ref{alg:estmoments} first computes an $(\eps_1,\delta/2)$-estimate $\wt{W}$ of the sum $W$ of the weights of elements in $A$ using the sum estimation algorithm in \cite{BT2022}\footnote{\cite{BT2022} states the sample complexity for probability of success at least $2/3$. Here, we are stating the bounds for an $(\eps_1,\delta/2)$-estimate. This is obtained using an application of the medians of means technique.}. Let the probability of sampling element $a_j$ using proportional sampling be given as $p_j=\frac{w(a_j)}{W}$. Note that we don't know $p_j$, however we can obtain a good approximation of it as follows. For a proportional sample $a_j$, we know its weight $w(a_j)$, and $\wt{W}$ gives us an approximation of $W$. Using this we get a approximation for $p_j$ as $\ti{p}_j=\frac{w(a_j)}{\wt{W}}$. Let $\E$ denote the event that $\wt{W}\in [(1-\eps_1)W,(1+\eps_1)W]$. In what follows we condition on event $\E$.

\begin{claim}\label{claim:prob} Conditioned on $\E$, for any $j\in [n]$, we have $\frac{p_j}{1+\eps_1}\leq \ti{p}_j\leq \frac{p_j}{1-\eps_1}$. \end{claim}
\begin{proof} Conditioned on $\E$, $(1-\eps_1)W\leq \wt{W}\leq (1+\eps_1)W$. The above inequality follows. \end{proof}

Conditioned on $\E$, for any $j$, we have $\frac{p_j}{1+\eps_1}\leq \ti{p}_j\leq \frac{p_j}{1-\eps_1}$. Given a proportional sample $a_j$, we define a random variable $X_j$ with value $\frac{w(a_j)^t}{\ti{p}_j}$. Then, we have $(1-\eps_1)S_t\leq \EE[X_j]\leq (1+\eps_1)S_t$. Here, $X_j$ is not an unbiased estimator of $S_t$. Next, we bound the variance of this estimator given as $\var[X_j]=\EE[X_j^2]-\EE^2[X_j]\leq \EE[X_j^2]$. 

\begin{align}
\EE[X_j^2]
& = \sum_{a_j\in A} \frac{w(a_j)^{2t}}{\ti{p}_j^2} p_j \nonumber\\
& \leq (1+\eps_1)^2 \sum_{a_j\in A} \frac{w(a_j)^{2t}}{p_j} \nonumber\\
& = (1+\eps_1)^2 \cdot W \cdot \sum_{a_j\in A} \frac{w(a_j)^{2t}}{w(a_j)} && \text{(Substituting $p_j$ with $\frac{w(a_j)}{W}$)} \nonumber\\
& = (1+\eps_1)^2 \cdot W \cdot \sum_{a_j\in A} w(a_j)^{2t-1} \label{eqn:upper-sample}
\end{align}

Let us obtain $l$ independent samples using proportional sampling and let these random variables be $X_1,\ldots,X_l$. Let $X=\frac{1}{l} \sum_{j=1}^l X_j$. We have $(1-\eps_1)S_t\leq \EE[X]=\EE[X_j]\leq (1+\eps_1)S_t$, and $\var[X]=\frac{\var[X_j]}{l}$. Using Chebyshev's inequality, we have $\Pr[|X-S_t|>\eps S_t]\leq \Pr[|X-\EE[X]|>(\eps-\eps_1) S_t]]\leq \frac{\var[X]}{(\eps-\eps_1)^2 S_t^2}$. For appropriate choice of parameter $l$, we show this probability to be at most a small constant using the following claim. 


\begin{claim} $\frac{\var[X]}{(\eps-\eps_1)^2 S_t^2}\leq \frac{(1+\eps_1)^2}{l(\eps-\eps_1)^2} \cdot n^{1-1/t}$. \end{claim}

\begin{proof}
\begin{align*}
& \frac{\var[X]}{(\eps-\eps_1)^2 S_t^2} \\
\leq &\frac{\EE[X_j^2]}{l(\eps-\eps_1)^2 \cdot S_t^2} \\
 \leq &\frac{(1+\eps_1)^2}{l(\eps-\eps_1)^2} \cdot \frac{W \cdot \sum_{a_j\in A} w(a_j)^{2t-1}}{S_t^2} && \text{(Using Equation (\ref{eqn:upper-sample}))}\\
 = & \frac{(1+\eps_1)^2}{l(\eps-\eps_1)^2} \cdot \frac{W \cdot ||w(A)||_{2t-1}^{2t-1}}{S_t^2} && \text{(vector $w(A)$ has length $n$)} \\
 \leq &\frac{(1+\eps_1)^2}{l(\eps-\eps_1)^2} \cdot \frac{W \cdot {(||w(A)||_t^t)}^{2-1/t}}{S_t^2} && \text{(Using Fact~\ref{fact:norms}, $||w(A)||_{2t-1}\leq ||w(A)||_t$)} \\
 = &\frac{(1+\eps_1)^2}{l(\eps-\eps_1)^2} \cdot \frac{||w(A)||_1 \cdot (||w(A)||_t^t)^{2-1/t}}{(||w(A)||_t^t)^2}\\
 = &\frac{(1+\eps_1)^2}{l(\eps-\eps_1)^2} \cdot \frac{||w(A)||_1}{||w(A)||_t}\\
 \leq &\frac{(1+\eps_1)^2}{l(\eps-\eps_1)^2} \cdot n^{1-1/t} && \text{(Using Fact~\ref{fact:norms}, $||w(A)||_1\leq n^{1-1/t} ||w(A)||_t$)}
\end{align*}
\end{proof}

Let $\eps_1=\eps/2$. For $l=48n^{1-1/t}/\eps^2$, this failure probability is at most $1/3$. Using the standard median trick, we show that for $v=48\log 2/\delta$, this failure probability can be reduced to be at most $\delta/2$. Let us define independent Bernoulli random variables $Z_1,\ldots,Z_v$ such that $\Pr[Z_i=1]=2/3$ for all $i$. Let $Z=\sum_{r=1}^v Z_r$. Now, conditioned on $\E$, the probability that the output of Algorithm \ref{alg:estmoments} does not lie in the interval $[(1-\eps)S_t,(1+\eps)S_t]$ is the same as the probability that the median of $Y_1,\ldots,Y_v$ lies outside the interval $[(1-\eps)S_t,(1+\eps)S_t]$. This probability is at most $\Pr[Z<v/2]$. Using a standard application of a Chernoff bound, given in Lemma \ref{lem:chernoff}, we have $\Pr[Z<v/2]\leq \delta/2$.

\paragraph{Correctness and Sample complexity bounds}

We need to show that Algorithm \ref{alg:estmoments} returns an estimate $ALG(A,t,\eps,\delta)$ for which with probability at least $1-\delta$, we have $(1-\eps)S_t\leq ALG(A,t,\eps,\delta)\leq (1+\eps)S_t$. Step $2$ of Algorithm \ref{alg:estmoments} uses $\Theta{(\frac{\sqrt{n}\log (2/\delta)}{\eps_1})}$ proportional samples to obtain an $(\eps_1,\delta/2)$ estimate $\wt{W}$ of $W$. Algorithm \ref{alg:estmoments} fails if either $\E$ does not hold or the estimate in Step $10$ is incorrect. Since both of these failure probabilities are at most $\delta/2$, the algorithm succeeds with probability at least $1-\delta$. 

Algorithm \ref{alg:estmoments} uses $\Theta{(\frac{\sqrt{n}\log (2/\delta)}{\eps_1})}$ proportional samples in Step 2 and uses $O(\frac{n^{1-1/t}\log 1/\delta}{\eps^2})$ proportional samples in Steps 3 and 4. Therefore, the required number of proportional samples is $O(\frac{\sqrt{n}\log 1/\delta}{\eps_1} + \frac{n^{1-1/t}\log 1/\delta}{\eps^2})$. For $\eps_1=\eps/2$, this gives $O(\frac{\sqrt{n}\log 1/\delta}{\eps} + \frac{n^{1-1/t}\log 1/\delta}{\eps^2})$.

\section{Lower bound for Moment Estimation using Proportional Sampling}\label{sec:lower-proportional}

We use Yao's minimax lemma to prove the sample complexity lower bound for obtaining an $(\eps,\delta)$ estimate of $S_t$ using any randomized algorithm. We construct two families of instances on which $S_t$ differs by at least a $(1\pm \eps)$-factor and show that it is hard to distinguish these two instances using a small number of proportional samples.

Our lower bound constructions are as follows. There are $n_1$ elements of weight $d_1$ and $n_2$ elements of weight $d_2$ in both the instances, where the values of the parameters $n_1, n_2$ and $d_1$ are the same in both instances, and the instances differ in the value of parameter $d_2$. The exact values of these parameters will be set below. In one instance we set $d_2=n$, where as for the other instance $d_2=0$. These choices for parameter values creates a gap of a multiplicative $(1\pm \eps)$ factor between the $S_t$ values of the two instances. One can differentiate these two instances only when an element of weight $d_2=n$ is sampled using proportional sampling, and we show that this requires a lot of samples. 

\begin{thm} For any $\eps,\delta \in (0,1)$ and $t>1$, any randomized algorithm that computes an $(\eps,\delta)$-estimate of $S_t$ requires $\Omega(\frac{n^{1-1/t}\ln 1/\delta }{\eps^2})$ proportional samples. \end{thm}

\begin{proof} We construct two families of instances which we show are hard to be distinguished using a few proportional samples by Yao's lemma. In the first instance we have $n_1$ elements with weight $d_1$ and $n_2$ elements of weight $0$, and in the second instance, there are $n_1$ elements of weight $d_1$ and $n_2$ elements of weight $d_2$, where the following values are used. The parameter values used in the lower bound constructions of the two instances are given as follows. 


\begin{gather}
\begin{align*}
\qquad \qquad \quad n_1 & =  \frac{n^2}{n + \eps^{\frac{2t-1}{t-1}}} & n_1 & =  \frac{n^2}{n+ \eps^{\frac{2t-1}{t-1}}}\\
\qquad \qquad \quad d_1 & =  n^{1-1/t} \eps^{1/(t-1)} & d_1 & =  n^{1-1/t} \eps^{1/(t-1)}\\
\qquad \qquad \quad n_2 & =  \frac{n \eps^{\frac{2t-1}{t-1}}}{n + \eps^{\frac{2t-1}{t-1}}} & n_2 & =  \frac{n \eps^{\frac{2t-1}{t-1}}}{n + \eps^{\frac{2t-1}{t-1}}}\\
\qquad \qquad \quad d_2 & =  0 & d_2 & = n
\end{align*}
\end{gather}

\noindent The $S_t$ value for the first instance is given as follows.
\begin{align*}
n_1 \cdot d_1^t + n_2 \cdot 0
& = \frac{n^2}{n + \eps^{\frac{2t-1}{t-1}}} \cdot n^{t-1} \eps^{t/(t-1)}\\
& = \frac{n^{t+1} \eps^{t/(t-1)}}{n + \eps^{\frac{2t-1}{t-1}}} 
\end{align*}

\noindent The $S_t$ value for the second instance is 
\begin{align*}
n_1\cdot d_1^t + n_2 \cdot d_2^t 
& = \frac{n^{t+1} \eps^{t/(t-1)}}{n + \eps^{\frac{2t-1}{t-1}}} +  \frac{n \eps^{\frac{2t-1}{t-1}}}{n + \eps^{\frac{2t-1}{t-1}}} \cdot n^t\\
& = \frac{n^{t+1} \eps^{t/(t-1)}}{n + \eps^{\frac{2t-1}{t-1}}} + \eps \frac{n^{t+1} \eps^{t/(t-1)}}{n + \eps^{\frac{2t-1}{t-1}}}\\
& = (1+\eps) \frac{n^{t+1} \eps^{t/(t-1)}}{n +\eps^{\frac{2t-1}{t-1}}}
\end{align*}

The above two instances differ in their $S_t$ values by a multiplicative factor of $(1+\eps)$. In order to distinguish these two instances, one is required to sample an element of weight $d_2=n$. Using proportional sampling, the probability of sampling an element of weight $n$ in the above instance is given to be at least

\begin{align*}
\frac{n_2 d_2}{n_2 d_2 + n_1 d_1}
& = \frac{\frac{n \eps^{\frac{2t-1}{t-1}}}{n + \eps^{\frac{2t-1}{t-1}}} \cdot n} {\frac{n \eps^{\frac{2t-1}{t-1}}}{n + \eps^{\frac{2t-1}{t-1}}} \cdot n + \frac{n^2 }{n + \eps^{\frac{2t-1}{t-1}}} \cdot n^{1-1/t} \eps^{1/(t-1)}}\\
& =  \frac{n^2 \eps^{\frac{2t-1}{t-1}}}{n^2 \eps^{\frac{2t-1}{t-1}} + n^2 \cdot n^{1-1/t} \cdot \eps^{\frac{1}{t-1}}}\\
& =  \frac{1}{1+\frac{n^{1-1/t}}{\eps^2}}
\end{align*}

Let $p=\frac{1}{1+\frac{n^{1-1/t}}{\eps^2}}$. The lower bound on the sample complexity for this instance distinguishing problem is given as the number of samples required to observe a \textit{success} with probability at least $(1-\delta)$ while drawing independent samples from $Geom(p)$. Here, $Geom(p)$ denotes a geometric distribution with success probability $p$. The number of samples required to observe one \textit{success} from $Geom(p)$ with probability at least $(1-\delta)$ is at least $\Omega(\frac{\ln 1/\delta}{p})$. Therefore, $\Omega(\frac{n^{1-1/t}\ln1/\delta}{\eps^2})$ samples are required to distinguish these two instances with probability at least $1-\delta$.

\end{proof}

\section{Estimation of Moments using Hybrid Sampling}\label{sec:lower-hybrid}

In this section we prove a lower bound showing that for the moment estimation problem, the hybrid sampling framework does not provide any significant advantage over access to just the proportional sampling oracle. In contrast, note that for the sum estimation problem, hybrid sampling-based algorithms in fact give much better sample complexity bounds over proportional sampling \cite{MPX2007,BT2022}. We prove the following result.

\begin{thm} For any $\eps,\delta>0$ and $t>1$, any algorithm having access to a hybrid sampling oracle requires to make at least $\Omega(\frac{n^{1-1/t}\ln 1/\delta}{\eps^2})$ queries to compute an $(\eps,\delta)$-estimate for $S_t$. \end{thm}

We show that the lower bound instance described in Section \ref{sec:lower-proportional} yields a lower bound for the hybrid sampling as well. In order to distinguish these instances, one is required to sample an element of weight $n$. We have seen that using just proportional sampling $\Omega(\frac{n^{1-1/t} \log 1/\delta}{\eps^2})$ samples are required. The probability of sampling an element of weight $d_2$ using uniform sampling is given as $\frac{n_2}{n_1+n_2}$. This probability using the values of the parameters from Section \ref{sec:lower-proportional} equals $\frac{n_2}{n_1+n_2} = \frac{1}{1+\frac{n}{\eps^{\frac{2t-1}{t-1}}}}$. The instances are distinguished if an element of weight $n$ is sampled using either proportional sampling or uniform sampling. These two probabilities are given as $\frac{1}{1+\frac{n^{1-1/t}}{\eps^2}}$ and $\frac{1}{1+\frac{n}{\eps^{\frac{2t-1}{t-1}}}}$, respectively. Overall, we get a lower bound of $\Omega(\min\{\frac{n^{1-1/t}}{\eps^2}, \frac{n}{\eps^{\frac{2t-1}{t-1}}}\}\ln 1/\delta)=\Omega(\frac{n^{1-1/t}\ln 1/\delta}{\eps^2})$ for the $(\eps,\delta)$ moment estimation problem using hybrid sampling.

\section{Characterization of Sample Complexity}\label{sec:characterize}

We define a {\it moment-density} parameter of the input that governs the sample complexity of the moment estimation problem using proportional sampling. For $S_t=\sum_{a\in A} w(a)^t$ and $W=\sum_{a\in A} w(a)$, we define the moment-density parameter as $$\rho=\max_{L\subseteq A} \frac{\frac{\sum_{a\in L} w(a)^t}{S_t}}{\frac{\sum_{a\in L}w(a)}{W}}=\max_{L\subseteq A} \frac{\sum_{a\in L} w(a)^t}{\sum_{a\in L} w(a)} \cdot \frac{W}{S_t}$$ 

We give an upper bound for an $(\eps,\delta)$-estimator of $S_t$ using $O(({\sqrt{n}}/{\eps}+\rho/\eps^2)\ln 1/\delta)$ proportional samples. 

\begin{thm}\label{thm:char-upper} There exists an algorithm $ALG$ that given proportional sampling access to the weights of elements of $A$ having moment-density parameter $\rho$ and parameters $t>1,\eps,\delta\in (0,1)$, provides an $(\eps,\delta)$-estimate of $S_t$ using $O(({\sqrt{n}}/{\eps}+\rho/\eps^2)\ln 1/\delta)$ proportional samples. \end{thm}

\begin{proof} Algorithm \ref{alg:estmoments} gives the above sample complexity bound. We adapt the calculations from Section \ref{sec:moments}. Following Equation \ref{eqn:upper-sample}, we can write the sample complexity of our $(\eps,1/3)$ estimator to be at most $O(\sqrt{n}/\eps + \frac{1}{\eps^2} \frac{W\cdot \sum_{a\in A} w(a)^{2t-1}}{S_t^2})$. We will show that $\frac{W\cdot \sum_{a\in A} w(a)^{2t-1}}{S_t^2} \leq \rho$.  
\begin{align*}
\frac{W\cdot \sum_{a\in A} w(a)^{2t-1}}{S_t^2}
& = \frac{W}{S_t} \cdot \frac{\sum_{a\in A} w(a)^{2t-1}}{\sum_{a\in A} w(a)^t}\\
& \leq \frac{W}{S_t} \cdot \max_{a\in A} \frac{w(a)^t}{w(a)}\\
&= \rho
\end{align*}

Therefore, the sample complexity of our $(\eps,1/3)$-estimate of $S_t$ is $O({\sqrt{n}}/{\eps}+\rho/\eps^2)$. The $(\eps,1/3)$-estimate of the moment can be improved to an $(\eps,\delta)$-estimate with a multiplicative $\ln 1/\delta$-factor sample complexity overhead using the standard median trick.
\end{proof}
Next, we show a $\Omega(\frac{\rho \ln 1/\delta}{\eps})$ lower bound on the sample complexity of any algorithm for the $(\eps,\delta)$-moment estimation problem on any instance with the moment-density parameter $\rho$. We construct two families of input distributions that are hard to be distinguished using a small number of proportional samples. The lower bound instances for the moment estimation using proportional sampling described in Section \ref{sec:lower-proportional} have the moment-density parameter values as $\rho=1$ and $\rho=O(n^{1-1/t}/\eps)$. Next, we show that similar lower bounds hold even when we restrict the input instances to have their moment-density parameter values to be within a constant factor.

\begin{thm}\label{thm:char-lower} For any $\rho$, $\eps,\delta>0$ and $t>1$, any randomized algorithm for an $(\eps,\delta)$-estimate of $S_t$ on an instance with the moment-density parameter $\rho$ requires $\Omega(\frac{\rho \ln 1/\delta}{\eps})$ proportional samples. \end{thm}
\begin{proof} 
  Given any $\rho,\eps,\delta$ as input, we construct two families of instances for which their moment-density parameter values $\rho_1,\rho_2=O(\rho)$ and their moment values $S_t$ differ by a $(1\pm \eps)$ factor. There are $n$ elements in both the instances. In the first instance $n_1$ elements will have weight $d_1$ and $n_2$ elements will have value $d_2$. In the second instance there are $n_1$ elements of weight $d_1$, $\frac{n_2}{3}$ elements of weight $d_2$, and $\frac{2n_2}{3}$ elements of weight $0$. We set the values of the parameters $n_1,n_2,d_1,d_2$ as follows:

\begin{gather}
\begin{align*}
\qquad \qquad \quad n_1 & =  \frac{n^2}{n +(3 \eps)^{\frac{2t-1}{t-1}}} & d_1 & =n^{1-1/t}( 3\eps)^{1/(t-1)}\\
\qquad \qquad \quad n_2 & =  \frac{n (3\eps)^{\frac{2t-1}{t-1}}}{n +(3 \eps)^{\frac{2t-1}{t-1}}} & d_2&=n
\end{align*}
\end{gather}

\noindent The moment value for the first instance is:
\begin{align*}
S_t = n_1d_1^t+n_2d_2^t&=\frac{n^2}{n +(3 \eps)^{\frac{2t-1}{t-1}}}\cdot n^{t-1}(3\eps)^{\frac{t}{t-1}}+\frac{n^{t+1} (3\eps)^{\frac{2t-1}{t-1}}}{n +(3 \eps)^{\frac{2t-1}{t-1}}} \\
&= \frac{n^{t+1}(3\eps)^{\frac{t}{t-1}}+n^{t+1}(3\eps)^{\frac{2t-1}{t-1}}}{n +(3 \eps)^{\frac{2t-1}{t-1}}}\\
&=\frac{n^{t+1}(3\eps)^\frac{t}{t-1}}{n+(3\eps)^{\frac{2t-1}{t-1}}}(1+3\eps)
\end{align*}

\noindent The moment value for the second instance is:
\begin{align*}
S_t = n_1d_1^t+\frac{n_2}{3}d_2^t+\frac{2}{3}n_2\cdot0 &=\frac{n^2}{n +(3 \eps)^{\frac{2t-1}{t-1}}}\cdot n^{t-1}(3\eps)^{\frac{t}{t-1}}+\frac{1}{3}\frac{n^{t+1} (3\eps)^{\frac{2t-1}{t-1}}}{n +(3 \eps)^{\frac{2t-1}{t-1}}}\\
&=\frac{n^{t+1}(3\eps)^\frac{t}{t-1}}{n+(3\eps)^{\frac{2t-1}{t-1}}}(1+\frac{3\eps}{3})\\
&=\frac{n^{t+1}(3\eps)^\frac{t}{t-1}}{n+(3\eps)^{\frac{2t-1}{t-1}}}(1+\eps)
\end{align*}

The moment values of the above two instances differ by a factor of $\frac{1+3\eps}{1+\eps}>1+\eps$. Let $\rho_1$ and $\rho_2$ be the moment-density parameters for these two instances, respectively. We compute the $\rho$ values of the instances as follows. 
\begin{align*}
\rho_1 & = \frac{n^t}{n} \cdot \frac{n_1 d_1 + n_2 d_2}{n_1 d_1^t+n_2 d_2^t}\\
& = n^{t-1} \cdot \frac{\frac{n^2}{n +(3 \eps)^{\frac{2t-1}{t-1}}} \cdot n^{1-1/t}( 3\eps)^{1/(t-1)} + \frac{n (3\eps)^{\frac{2t-1}{t-1}}}{n +(3 \eps)^{\frac{2t-1}{t-1}}}\cdot n }{\frac{n^2}{n +(3 \eps)^{\frac{2t-1}{t-1}}} \cdot n^{t-1}( 3\eps)^{t/(t-1)} + \frac{n (3\eps)^{\frac{2t-1}{t-1}}}{n +(3 \eps)^{\frac{2t-1}{t-1}}}\cdot n^t}\\
&=\frac{n^{1-1/t}+(3\eps)^2}{3\eps+(3\eps)^2}=\frac{n^{1-1/t}+9\eps^2}{3\eps+9\eps^2}
\end{align*}

Similarly, we compute $\rho_2$ as
\begin{align*}
\rho_2 & =\frac{n^t}{n} \cdot \frac{n_1 \cdot d_1 + n_2/3 \cdot d_2 + 2n_2/3 \cdot 0}{n_1 \cdot d_1^t + n_2/3 \cdot d_2^t + 2n_2/3 \cdot 0}\\ 
& = \frac{n^{1-1/t}+\frac{1}{3}(3\eps)^2}{3\eps(1+\eps)}= \frac{n^{1-1/t}+3\eps^2}{3\eps+3\eps^2} 
\end{align*}

We have two instances as given above where their $\rho$ values are within a constant factor of each other and the moment values differ by at least a factor of $(1+\eps)$. These instances using proportional sampling remain indistinguishable until an element with weight $d_2$ is sampled. The probability of sampling an element of weight $d_2$ using proportional sampling is at most:
\begin{align*}
\frac{n_2d_2}{n_1d_1+n_2d_2}=\frac{1}{1+\frac{n_1d_1}{n_2d_2}}=\frac{1}{1+\frac{n^2n^{1-1/t}(3\eps)^{1/(t-1)}}{n^2(3\eps)^{\frac{2t-1}{t-1}}}}=\frac{1}{1+\frac{n^{1-1/t}}{9\eps^2}}\leq \frac{1}{1+O(\frac{\rho}{\eps})}
\end{align*} 

Therefore, we require at least $\Omega(\rho/\eps)$ samples to distinguish the two instances with constant probability of success. Any algorithm giving an $(\eps,\delta)$-estimate would need to make $\Omega (\frac{\rho \ln 1/\delta}{\eps})$ proportional samples.
\end{proof}

\section{Moment estimation for $0<t<1$}\label{sec:small-t}

\subsection{Moment estimation using Proportional Sampling for 1/2<t<1} We use Algorithm \ref{alg:estmoments} to estimate moment $S_t$ for $1/2<t<1$. We adapt the analysis in Section \ref{sec:moments} to prove the following result.

\begin{thm} There exists an algorithm $ALG$ that given proportional sampling access to the weights of the elements of a set $A$ and parameters $1/2<t<1$, $\eps\in (0,1)$, provides an $(\eps,1/3)$-estimate of $S_t$ using $O(\frac{\sqrt{n}}{\eps}+\frac{n^{\frac{1}{t}-1}}{\eps^2})$ samples. \end{thm}

The sample complexity of the algorithm for $1/2<t<1$ is given as follows. Note that the sample complexity of our algorithm is given as $\frac{W (||w(A)||_{2t-1})^{2t-1}}{\eps^2 (||w(A)||_t)^{2t}}$. We can rewrite it as $\frac{||w(A)||_1}{\eps^2 ||w(A)||_t} \cdot (\frac{||w(A)||_{2t-1}}{||w(A)||_t})^{2t-1}$. Using Fact \ref{fact:norms}, this is upper bounded as $n^{\frac{1}{t}-1}/\eps^2$. Overall sample complexity of the $(\eps,1/3)$-estimator of the moment is $O(\sqrt{n}/\eps+ n^{\frac{1}{t}-1}/\eps^2)$. Using the median trick we transform the $(\eps,1/3)$-estimator to an $(\eps,\delta)$-estimator with sample complexity $O((\sqrt{n}/\eps+ n^{\frac{1}{t}-1}/\eps^2) \ln 1/\delta)$.

\subsection{Lower bound for $t\leq 1/2$} We show that there are no sublinear algorithms for the problem when $t\leq 1/2$.

\begin{thm} For any $\eps>0$ and $t\leq 1/2$, any randomized algorithm that computes an $(\eps,1/3)$-estimate of $S_t$ requires $\Omega(n)$ proportional samples. \end{thm}
\begin{proof} We construct two families of instances which we show are hard to be distinguished using proportional sampling. In one instance we have one element with weight $(n-1)$, and the rest $(n-1)$ elements with weight $0$. In another instance, we have one element of weight $(n-1)$, and $(n-1)$ elements have weight $\frac{\eps^{1/t}}{n-1}$. The moment values of the two instances are given as $(n-1)^t$ and $(n-1)^t + \eps (n-1)^{1-t}\geq (1+\eps) (n-1)^t$ as $t\leq 1/2$. Hence, there exists a $(1+\eps)$-multiplicative gap between the moment values of these two instances.

These two instances are distinguished using proportional sampling once we sample an element of weight $\frac{\eps^{1/t}}{n-1}$. The probability of sampling such an element is at most $\frac{\eps^{1/t}}{\eps^{1/t} + (n-1)}$. Hence, one requires $\Omega(n)$ samples to distinguish the two instances with constant probability. \end{proof}

\section{Conclusion} In this paper we improve the results for moment estimation using proportional sampling and proved optimal sample complexity bounds for moment estimation using proportional sampling for $t\geq 2$. Since hybrid sampling does not provide better guarantees for the moment estimation problem, one possible approach to obtain better bounds might be to explore whether conditional sampling-based approaches give better results. Conditional sampling has been used to design algorithms for the sum estimation problem \cite{ACK2015} but to the best of our knowledge no such algorithm is known for the moment estimation problem.

\subsubsection{Acknowledgements} Anup Bhattacharya is supported by the Science and Engineering Research Board (SERB) via the project (CRG/2023/002119). The authors thank the reviewers for suggestions to improve the manuscript. 

\bibliographystyle{splncs04}
\bibliography{ref}

\end{document}